\documentclass[11pt]{article}
\usepackage{graphicx,amsmath,amssymb,amsthm}
\usepackage{fullpage}
\usepackage{color}

\usepackage{epstopdf}
\usepackage{epsfig}
\usepackage{tikz}
\usepackage{subfig}

\usepackage{url}

\newcommand{\perpslab}[1]{\mathrm{slab}(#1)}
\newcommand{\etal}{{\em et al.~}}

\newtheorem{theorem}{Theorem}
\newtheorem{lemma}[theorem]{Lemma}
\newtheorem{corollary}[theorem]{Corollary}

\newcommand{\changed}[1]{#1}
\newcommand{\ignore}[1]{}
\newcommand{\changedagain}[1]{{#1}}
\newcommand{\changeA}[1]{{#1}}
\newcommand{\changeAL}[1]{{#1}}

\newcommand{\changeS}[1]{{#1}}
\newcommand{\changeSH}[1]{{#1}}

\begin{document}

\title{ Self-Approaching Graphs}

\author{
Soroush Alamdari\thanks{Cornell University, Ithaca, USA {\tt alamdari@cs.cornell.edu}}
\and
Timothy M. Chan\thanks{Cheriton School of Computer Science, University of Waterloo, Waterloo, Canada
    {\tt \{tmchan, alubiw, vpathak\}@uwaterloo.ca}}
\and
Elyot Grant\thanks{Massachusetts Institute of Technology, Cambridge, USA {\tt elyot@mit.edu}}
\and
Anna Lubiw\footnotemark[2]\and
Vinayak Pathak\footnotemark[2]
}

\maketitle

\begin{abstract}
In this paper we introduce \emph{self-approaching}  graph drawings. A straight-line drawing of a graph is \emph{self-approaching} if, for any origin vertex $s$ and any destination vertex $t$, there is an $st$-path in the graph such that,
for any point $q$ on the path,
as a point $p$ moves continuously along the path from the origin to $q$, the  Euclidean distance  from $p$ to $q$ is always decreasing.
This is a more stringent condition than a greedy drawing (where only the distance between vertices on the path and the destination vertex must decrease), and guarantees that the drawing is a 5.33-spanner.

We study three topics:
(1) recognizing self-approaching drawings;
(2) constructing self-approaching drawings of a given graph;
(3) constructing a self-approaching Steiner network connecting a given set of points.

We show that: (1) there are efficient algorithms to test if a polygonal path is self-approaching in $\mathbb{R}^2$ and $\mathbb{R}^3$,
but it is NP-hard to test if a given graph drawing in $\mathbb{R}^3$ has a self-approaching $uv$-path;
(2)
we can characterize the trees that have self-approaching drawings;
(3) for any given set of terminal points in the plane, we can find a linear sized network that has a self-approaching path between any ordered pair of terminals.


\end{abstract}

\section{Introduction}
A straight-line graph drawing (or ``geometric graph'') in the plane has points for vertices, and straight line segments for edges, where the weight of an edge is its Euclidean length.  The drawing need not be planar.
Rao \etal\cite{Rao:GeoRouting:2003} introduced the idea of greedy drawings.  A \emph{greedy drawing} of a graph is a straight-line drawing in which, for each origin vertex $s$ and destination vertex $t$, there is a neighbor of $s$ that is closer to $t$ than $s$ is, i.e., there is a \emph{greedy} $st$-path $P=(s=p_1,p_2,\ldots,p_k = t)$ such that the Euclidean distances $D(p_i, t)$ decrease as $i$ increases.
This idea has 
attracted great interest in recent years (e.g. \cite{Angelini:2009,Bose:theta6:2012,Goodrich:2008,He:2011,Leighton:2010,Papadimitriou:2005}) mainly because a greedy drawing of a graph permits greedy local routing.

It is a very natural and desirable property that a path should always get closer to its destination, but there is more than one way to define this.
Although every vertex along a greedy path gets closer to the destination, the same is not true of intermediate points along edges.
See Figure~\ref{fig:greedy-vs-SA}.

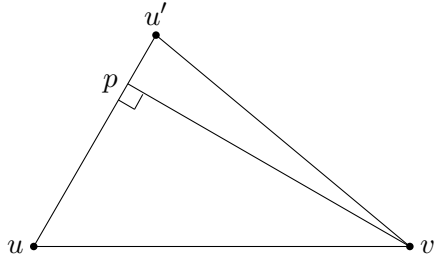
\begin{figure}
\centering

\begin{tikzpicture}[scale=2.5]
\draw (0,0) -- (2, 0);
\draw (0.5, 0.866) node[left] {$p$} -- (2,0);
\draw (0,0) -- (0.65, 1.125);
\draw (2, 0) -- (0.65, 1.125);

\fill[black] (0, 0) node[left] {$u$} circle (0.02);
\fill[black] (2, 0) node[right] {$v$} circle (0.02);
\fill[black] (0.65, 1.125) node[above] {$u'$} circle (0.02);

\draw (0.45, 0.78) -- (0.536, 0.73);
\draw (0.536, 0.73) -- (0.58, 0.81);
\end{tikzpicture}

\caption{As we move from $u$ towards $u'$, distance to $v$ first decreases (until $p$), then increases. However, $D(u', v) < D(u, v)$.}
\label{fig:greedy-vs-SA}
\end{figure}


Another disadvantage of greedy paths is that the
length of a greedy path is not bounded in terms of the Euclidean distance between the endpoints.
This is another natural and desirable property for a path to have, and is captured by the
\emph{dilation} (or ``stretch factor'') of a graph drawing---the maximum, over vertices $s$ and $t$, of the ratio of their distance in the graph to their Euclidean distance.
The dilation factor of greedy graph drawings can be unbounded.

Icking \etal\cite{Icking:self-approachingcurves:1995} introduced a stronger notion of
 ``getting closer'' to a destination, that addresses both shortcomings of greedy paths.
A curve from $s$ to $t$ is \emph{self-approaching} if for any three points $a, b, c$ appearing in that order along the curve, we have $D(a,c) \ge D(b,c)$.
Icking \etal proved that a self-approaching curve has \emph{detour} at most 5.3332, where the \emph{detour} or \emph{geometric dilation} of a curve is the supremum over points $p$ and $q$ on the curve, of the ratio of their distance along the curve to their Euclidean distance $D(p,q)$.
This is stronger than dilation in that we consider all pairs of points, not just all pairs of vertices.

In this paper we introduce the notion of a \emph{self-approaching} graph drawing---a straight-line drawing  that contains, for every pair of vertices $s$ and $t$,
a self-approaching $st$-path and a self-approaching $ts$-path (which need not be the same).
We also explore the related notion of an \emph{increasing-chord} graph drawing, which has the stronger property that every pair of vertices is joined by a path that is self-approaching in both directions.
Rote~\cite{Rote:ICcurves:1994} proved that increasing-chord paths have geometric dilation  at most 2.094.

Our first result is a linear time algorithm to recognize a self-approaching polygonal path in the plane.
This extends to $\mathbb{R}^3$, with some slow-down---we give an algorithm that runs in time $O(n \log^2 n / \log \log n)$ and a lower bound of $\Omega(n \log n)$.  This is in Section~\ref{sec:testingPaths}.

We do not know the complexity of recognizing self-approaching graph drawings in the plane or higher dimensions.
One approach would be to find, for every pair of vertices $u$ and $v$, a self-approaching path from $u$ to $v$ in the graph drawing.   This problem is open in $\mathbb{R}^2$ but we show that it  is NP-hard in $\mathbb{R}^3$.  This is in Section~\ref{sec:SAPathsInGraphs}.
%

Next, we consider the question of constructing a self-approaching drawing for a given graph.  We give a linear time algorithm to recognize the trees that have self-approaching drawings.  See Section~\ref{sec:SADrawability}.

Finally,  we consider the problem of connecting a given set of terminal points in the plane by a network that has a self-approaching path between every pair of terminals.  We show that this can be done with a linear sized network.  See Section~\ref{sec:SAspanners}.
\section{Background}

A \emph{spanner} is a graph of bounded dilation.  Spanners have been very well-studied---see for example the book by Narasimhan and Smid~\cite{Spanners} and the survey by Eppstein~\cite{Eppstein:Spanners:2000}.
A main goal is to efficiently construct a spanner on a given set of points, with the objective of minimizing dilation while keeping
 the number or total length of edges small.  For recent results, see, e.g.,~\cite{Aronov:Spanners:2008,Giannopoulos:hard-spanners:2010}.
 If Steiner vertices are allowed, their number should also be minimized, and different versions of the problem arise if we include the Steiner points in measuring the dilation, see~\cite{Ebbers:Dilation:2007}.

The \emph{detour} of a graph drawing is defined to be the supremum, over all points $p,q$ of the drawing (whether at vertices, or interior to edges) of the ratio of their distance in the graph to their Euclidean distance.  Note  that if two edges cross in the drawing, then the detour is infinite.  By contrast,  a self-approaching drawing may have crossing edges, for example, any complete geometric graph is self-approaching.
Constructing a network to minimize detour has also been considered~\cite{Ebbers:detour:2006,Dumitrescu:2009}, though not as extensively as spanners.

Relevant background on greedy drawings is as follows.
Answering a conjecture of Papadimitriou and Ratajczak~\cite{Papadimitriou:2005}, Leighton and Moitra~\cite{Leighton:2010}
\changeS{and Angelini \etal\cite{Angelini:2009} independently showed that any 3-connected planar graph has a greedy drawing. However, the number of bits needed for the coordinates in these embeddings is large for routing purposes. Goodrich and Strash~\cite{Goodrich:2008} showed how to find a greedy path in such drawings without storing the actual coordinates, but instead using local information of small size.}
Moitra \cite{Moitra:Thesis:2009} used combinatorial conditions to classify the trees that have greedy embeddings and \changeSH{very recently N{\"o}llenburg and Prutkin \cite{Nollenburg:2013} completely characterized greedy drawable trees.}
Connecting the ideas of greedy drawings and spanners,
Bose \etal\cite{Bose:theta6:2012} showed that every triangulation has an embedding in which local routing produces a path of bounded dilation.

Self-approaching drawings are related to \emph{monotone drawings} in which, for every pair of vertices $s$ and $t$,  there is an $st$-path that is monotone in some direction.  This concept was introduced by Angelini, et al.,~\cite{Angelini:MonoDraw:2012} who
gave algorithms to construct monotone planar drawings of trees and planar biconnected graphs.   A follow-up paper~\cite{Angelini:MonoFixed:2011} considers the case where a planar embedding is specified.  Self-approaching drawings are not necessarily monotone, and monotone drawings are not necessarily self-approaching.   The one relationship is that any increasing-chord drawing is  a monotone drawing.


Although a monotone path need not be self-approaching, there is a stronger condition that does imply self-approaching, namely that
the path is monotone in both the $x$- and $y$-directions.
Thus, a network with an $xy$-monotone path between every pair of terminals is a self-approaching network.  A \emph{Manhattan network} has horizontal and vertical edges and includes an $L_1$ shortest path between every pair of terminals.   So a Manhattan network is self-approaching.
There is considerable work on finding Manhattan networks of minimum total length (so-called ``minimum Manhattan networks").  There are efficient algorithms with approximation factor 2, and the problem has been shown to be NP-hard~\cite{Chin:MMN-NP-complete:2011}.
More relevant to us is the result of Gudmundsson et al.~\cite{Gudmundsson:smallManhattan:2007} that every point set admits a Manhattan network of $O(n \log n)$ vertices and edges, and there are point sets for which any Manhattan network has size at least $\Omega(n \log n)$.
{\changed This contrasts with our result that every point set admits a self-approaching network of linear size.}

For results on computing the dilation or detour of a path or graph, see the survey by
 Gudmundsson and Knauer~\cite{Gudmundsson:Dilation:2007} and the paper by Wulff-Nilsen~\cite{Wulff-Nilsen:Detour:2010}.

The Delaunay triangulation has several good properties that are relevant to us: it has dilation factor below 2~\cite{Xia:Delaunay:2011}, and is a greedy drawing~\cite{Bose:routing:2004}, although greedy paths in a Delaunay triangulation do not necessarily have bounded dilation.
It is natural to conjecture that the Delaunay triangulation is self-approaching, but we show that this is not the case.


\section{Preliminaries}
We let $D(u,v)$ denote the Euclidean distance between points $u$ and $v$ in $\mathbb{R}^d$.
{\changed Formally, a
\emph{curve} is a continuous function $f\colon[0,1]\rightarrow \mathbb{R}^d$, and an $st$\emph{-curve} is a curve $f$ with $f(0)=s$ and $f(1)=t$.
The \emph{reverse curve} is $f(1-t), t \in [0,1]$.
For convenience, we will identify a curve with its image, and ignore the particular parameterization.  When we speak of points $a$ and $b$ \emph{in order along the curve}, or with $b$ \emph{later than} $a$ on the curve, we mean that $a=f(t_1)$ and $b=f(t_2)$ for some $0 \le t_1 \le t_2 \le 1$.
A curve is \emph{self-approaching} if for any three points $a,b,c$ in order along the curve, we have $D(a,c)\geq D(b,c)$. See Figure~\ref{fig:exampleSAICcurve}(a).
%
Note that this definition is sensitive to the direction of the curve---it may happen that a curve is self-approaching but its reverse is not.}

{\changed A curve has \emph{increasing chords} if for any four points $a,b,c,d$ in order along the curve we have $D(a,d) \ge D(b,c)$. See Figure~\ref{fig:exampleSAICcurve}(b) for an example.
Note that if a curve has increasing chords then the reverse curve also has increasing chords, and the curve and its reverse are both self-approaching.
The converse also holds: if a curve and its reverse are both self-approaching then the  curve has increasing chords, as we then have $D(a,d) \ge D(a,c) \ge D(b,c)$ for any points $a,b,c,d$ in order along the curve.}

\begin{figure}[htbp]
\begin{center}
\hspace*{\fill}
\subfloat[]{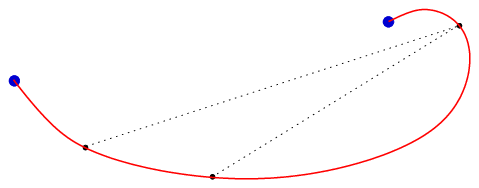}
\hspace*{\fill}
\subfloat[]{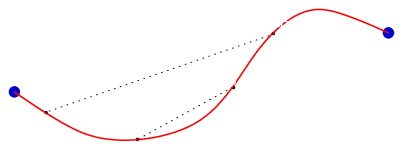}
\hspace*{\fill}
\end{center}
\caption{(a) A self-approaching $st$-curve and (b) an increasing-chord curve in $\mathbb{R}^2$.}
\label{fig:exampleSAICcurve}
\end{figure}



The following characterization of self-approaching curves is straightforward:
\begin{lemma}(\cite{Icking:self-approachingcurves:1995})
\label{lem:perpendSA}
{\changed A piecewise-smooth curve is self-approaching iff for each point $p$ on the curve, the line perpendicular to the curve at $p$ does not intersect the curve at a later point.}
\end{lemma}


\begin{corollary}
\label{lem:perpendIC}
{\changed A piecewise-smooth curve has increasing chords iff each line perpendicular to the curve intersects the curve at no other point.}
\end{corollary}

When dealing with straight-line drawings of graphs, we apply Lemma~\ref{lem:perpendSA} to piecewise-linear curves.
\changeA{
For distinct points $u$ and $v$, let $\overline{uv}$ be the line passing through $u$ and $v$.
See Figure~\ref{fig:ell}.
Let $l_{uv}$ denote the line that passes through $v$ and is perpendicular to $\overline{uv}$, noting that $l_{uv}$ and $l_{vu}$ are distinct parallel lines.
Let $l_{uv}^+$ denote the closed half-plane that has boundary $l_{uv}$ and does not contain $u$, and define $l_{vu}^+$ similarly.
Let $\perpslab{uv}$ be the open strip bounded by $l_{uv}$ and $l_{vu}$, in other words,
the complement of $l_{uv}^+ \cup l_{vu}^+$.
With this notation, we can restate the lemma as follows:}

\begin{figure}[htb]
\centering
\includegraphics[width=2in]{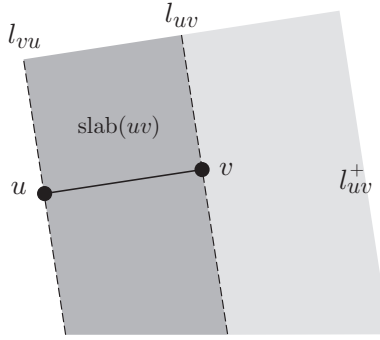}
\caption{
\changeA{The lines $l_{vu}$ and $l_{uv}$, the open $\perpslab{uv}$ (darkly shaded), and the closed half-plane $l_{uv}^+$ (lightly shaded).}
}
\label{fig:ell}
\end{figure}

\begin{corollary}
\label{cor:edgeperpend}
Let $P=(v_1,v_2,\ldots,v_n)$ be a directed path embedded in $\mathbb{R}^2$ via straight line segments.  Then, $P$ is self-approaching iff for all $1 < i < j \leq n$, the point $v_j$ lies in $l_{v_{i-1}v_i}^+$.  Equivalently, $P$ is self-approaching iff for all $1 < i \leq n$, the convex hull of $\{v_i,v_{i+1},\ldots,v_n\}$ lies in $l_{v_{i-1}v_i}^+$.
\end{corollary}


Analogous characterizations are also possible in higher dimensions, with the half-planes $l_{v_{i-1}v_i}^+$ replaced by half-spaces bounded by hyperplanes orthogonal to $\overline{v_{i-1}v_i}$.


\section{Testing whether paths are self-approaching}
\label{sec:testingPaths}
Corollary~\ref{cor:edgeperpend} implicitly suggests an algorithm to determine whether a directed path embedded in 
Euclidean space is self-approaching.  In this section, we provide \ignore{improved} algorithms for this task in two and three dimensions, as well as a lower bound.  We assume a real RAM model in which all simple geometric operations can be performed in $O(1)$ time, and we assume that a straight-line drawing of a path $P=(v_1,v_2,\ldots,v_n)$ is represented explicitly as a list of $n$ points (requiring $O(n)$ space).
\begin{theorem}
\label{thm:pathSolver}
Given a straight-line drawing of a path $P=(v_1,v_2,\ldots,v_n)$ in the plane, it is possible to test whether $P$ is self-approaching in linear time.
\end{theorem}
\begin{proof}
By Corollary~\ref{cor:edgeperpend}, we must only check that for all $1<i\leq n$, the convex hull of $\{v_{i},\ldots , v_n\}$ lies in $l_{v_{i-1}v_i}^+$.  We can do all of these checks in $O(n)$ time by performing them iteratively, beginning with $i=n$ and processing the points in decreasing order.  While doing this, we will either show that $P$ is not self-approaching, or we will be able to use the properties of self-approaching paths to construct the convex hull of the traversed vertices incrementally in linear total time by an algorithm similar to Graham's scan~\cite{Graham}.

We now describe a step of the algorithm.  Assume that the directed path $P_i = \{v_i,\ldots,v_n\}$ is self-approaching and assume the convex hull $C$ of vertices \{$v_i,\ldots , v_n$\} has already been computed and is stored by keeping track of the neighbors of each vertex on its boundary.  Since $P_i$ is self-approaching, point $v_i$ must lie on the boundary of $C$ (by Corollary~\ref{cor:edgeperpend}).  Let $v_i^1$ and $v_i^2$ be the neighbors of $v_i$ in $C$.
\changeA{Note that $C$ lies in $l_{v_{i-1}v_i}^+$ if and only if it does not intersect $\perpslab{v_{i-1}v_i}$ and that happens if and only if the line segments $\overline{v_iv_i^1}$ and $\overline{v_iv_i^2}$ do not intersect $\perpslab{v_{i-1}v_i}$.}
We can check this in $O(1)$ time.
If an intersection is found, then $P$ is not self-approaching and we can terminate the algorithm.  Otherwise, we add $v_{i-1}$ to $C$ and recompute the convex hull.  This can be done by repeatedly removing the vertices of $C$ on both sides of $v_i$ until convex angles are obtained\ignore{, after which the convex hull of $C \cup \{v_{i-1}\}$ will remain}.  Each vertex in $P$ will be removed at most once from a convex hull in some step of the algorithm, so the total running time for all steps of the algorithm is $O(n)$.
\end{proof}


In three dimensions, we can obtain a similar result with slightly worse running time using an existing convex hull data structure that supports point insertion and half-space range emptiness queries.
\begin{theorem}
\label{thm:pathSolver3D}
Given a straight-line drawing of a path $P=(v_1,v_2,\ldots,v_n)$ in $\mathbb{R}^3$, it is possible to test whether $P$ is self-approaching in $O(n \log^2 n/\log\log n)$ time.
\end{theorem}
\begin{proof}
The proof is analogous to that of Theorem~\ref{thm:pathSolver}, with the only change being that we must employ a more complicated data structure to store the convex hull and test whether it intersects a given half-space range.  For each edge $v_{i-1}v_i$, we can ensure that $\perpslab{v_{i-1}v_i}$ does not intersect the convex hull $C$ by performing two half-space range emptiness queries on $C$.  If no intersection is found, then we may insert point $v_{i-1}$ to our data structure and perform the next iteration of the algorithm.  If the algorithm successfully inserts all points into $C$, then the path $P$ must be self-approaching.

Achieving the stated running time requires a nontrivial data structure combining several known ideas.  There is a static data structure for half-space range emptiness in $\mathbb{R}^3$ with $O(n)$ space and $O(\log n)$ query time, by reduction to {\em planar point location\/} in dual space \cite{Kir:pl}; the preprocessing time is $O(n)$ if we are given the convex hull. The static data structure can be transformed into a semidynamic data structure with $O(b\log_b n)$ amortized insertion time and $O(\log_b n\log n)$ query time for a given parameter $b$, by known techniques---namely, a $b$-ary version of Bentley and Saxe's {\em logarithmic method\/}~\cite{BenSax}, using Chazelle's linear-time algorithm for merging two convex hulls~\cite{Cha:merge} as a subroutine. By setting $b=\log n$, both amortized insertion time and query time are bounded by $O(\log^2 n/\log\log n)$, yielding the desired result.
\end{proof}

Next, we show that Theorem~\ref{thm:pathSolver3D} is tight up to a factor of $\log n/\log\log n$ by proving a lower bound of $\Omega(n \log n)$ on the running time of any algorithm for determining whether a directed path embedded in $\mathbb{R}^3$ is self-approaching.  We do this by reducing from the \emph{set intersection problem}, for which a solution requires $\Omega(n \log n)$ time on an input of size $n$ in the algebraic computation tree model \cite{Ben83}.  We can show the following:
\begin{theorem}
\label{thm:hardnesspath}
Given a straight-line drawing of a path $P=(v_1,v_2,\ldots,v_n)$ in $\mathbb{R}^3$, at least $\Omega(n\log n)$ time is required in the algebraic computation tree model to test whether $P$ is self-approaching.
\end{theorem}
\begin{proof}
We first need a few gadgets for our reduction. Let $\beta = \pi/6$ and $\alpha = 1$.  For a point $p \in \mathbb{R}^2$, we define a \emph{cannon} $c$ at $p$ to be an embedding of a 3-vertex path $[c^0,c^1,c^2]$ where the points are located as follows:
\begin{itemize}
\item $c^0$ is placed at $p$,
\item $c^2$ is placed at $p + (1,0)$, that is, $\alpha$ units to the right of $p$, and
\item $c^1$ is placed at $p + (3/4, \sqrt{3}/4)$, on the line that meets the $x$-axis at an angle $\beta$ and passes through $c^0$, such that the angle
$\angle{c^0c^1c^2}$
is a right angle.
\end{itemize}
Similar to a cannon, a \emph{target} $t$ at point $p$ with respect to \changedagain{a line} $\ell$ is an embedding of a 3-vertex path $[t^0,t^1,t^2]$, where the points in $t$ are positioned as follows:
\begin{itemize}
\item $t^0$ is placed at $p$,
\item $t^1$ at the intersection of $\ell$ and $\ell^\prime$, where $\ell^\prime$ is the line of slope 1 passing through $t^0$, and
\item $t^2$ is placed on the $x$-axis such that the angle
$\angle{t^0t^1t^2}$   
is a right angle.
\end{itemize}

With these gadgets in hand, we now present a reduction from the set intersection problem.  Let $\mathcal{I}$ be an instance of the set intersection problem, where we are asked to check if there is a common element in sets $A$ and $B$. Using Yao's improvement to Ben-Or's lower bound constructions for algebraic computation trees \cite{Yao91}, it suffices to consider the case where $A$ and $B$ are sets of non-negative integers.  Letting $M$ be the maximum element in $A$ and $B$, we first divide each element of $A$ and $B$ by $2M/\pi$ so that both $A$ and $B$ are subsets of $[0,\pi/2]$, noting that this can be done in linear time.  Let $\varepsilon < \pi/2M$ so that $|a-b|>\varepsilon$ for all $a,b\in A \cup B$ with $a\neq b$, and let $\gamma$ be a sufficiently large constant (depending on $M$).  Using \changedagain{the elements} of $A$ and $B$, we embed a path $P=\{v_0,v_1,v_2,\ldots,v_{1+2|A|+2|B|}\}$ in $\mathbb{R}^3$ as follows:
\begin{enumerate}
\item Start with the vertex $v_0$ placed \changedagain{at} the origin.
\item For each $1 \leq i \leq |A|$, place a cannon $c_i$ in the $xy$-plane, attached to the current path, with $c_1^0 = v_0$ and $c_i^0 = c_{i-1}^2$ for $i \geq 2$.  Cannon $c_i$ \changedagain{represents} the element $a_i\in A$.  At this stage, the path should appear as a chain of $|A|$ cannons lined up along the $x$-axis.
\item Place the next vertex $v_{2|A|+1}$ of the path \changedagain{at} $(\alpha|A|+\gamma,0)$.
\item For each $1\leq i \leq |B|$, add a target $t_i$ in the $xy$-plane, placed at the end of the current path with respect to $\ell = \overline{v_0v_1}$.  Target $t_i$ \changedagain{represents} the element $b_i \in B$ and the targets, like the cannons, \changedagain{are} aligned along the $x$-axis. Figure~\ref{fig:reductionPath} shows what the path looks like at this point.
\item\label{enum:rotate1} Modify the embedding by rotating each cannon about the $x$-axis through an angle $a_i$ (in other words, relocate $p_i^1$ from $(x,3/4,0)$ to $(x,3/4\cos(a_i),3/4\sin(a_i))$).
\item Similarly, rotate each target $t_i^1$ about the $x$-axis through an angle $b_i$ by relocating $t_i^1$.
\item Let $P$ be the path obtained after these rotations.
\end{enumerate}

\begin{figure}[ht]
\begin{center}
\resizebox{6in}{!}{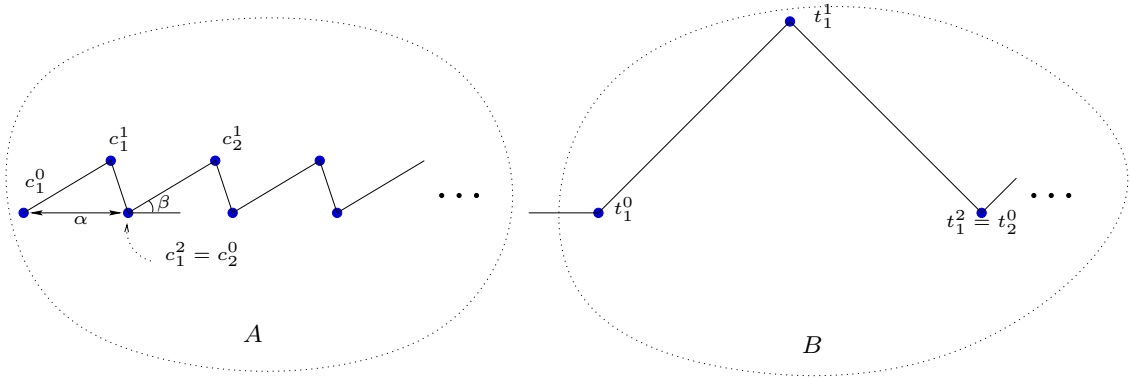}
\caption{The cannons (left) and the targets (right).}
\label{fig:reductionPath}
\end{center}

\end{figure}
\begin{figure}[ht]
\begin{center}
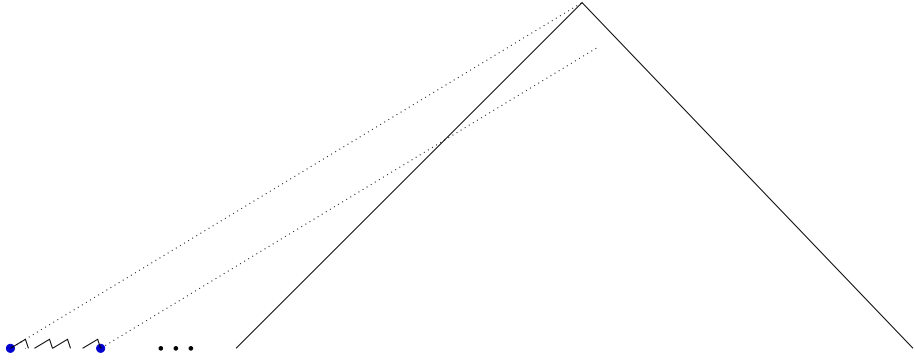
\caption{Placement of a target.}
\label{fig:reductionPath2}
\end{center}
\end{figure}

\changedagain{Our proof is based on the claim that $P$ is} a self-approaching path (in the $v_0$ to $v_{1+2|A|+2|B|}$ direction) if and only if $A$ and $B$ do not intersect. More specifically, $\perpslab{c_i^1c_i^2}$ collides with the target $t_j$ if and only if element $a_i$ equals element $b_j$.

\emph{Only if:} Assume $a_i=b_j$.  It is then easy to see that $\perpslab{c_i^1c_i^2}$ collides with the target $t_j$, since both the cannon $c_i$ and the target $t_j$ are rotated around the $x$-axis through the same angle.  It follows, by Lemma~\ref{lem:perpendIC}, that $P$ is not self-approaching.

\emph{If:} By Lemma~\ref{lem:perpendIC}, it suffices to show that if $A$ and $B$ do not intersect, then for any edge $e$ in $P$, $\perpslab{e}$ does not intersect any edges in the path after $e$. It is straightforward from our construction that the only way such an intersection can occur is if $\perpslab{c_i^1c_i^2}$ intersects a point $t_j^1$ for some $i$ and $j$.  Let $s$ be $\perpslab{c_{|A|}^1c_{|A|}^2}$ as it is positioned prior to step \ref{enum:rotate1} in the construction.
Define $\theta$ to be the minimum amount that we need to rotate the target $t_1$, so that the point $t_1^1$ does not lie in $s$. It is easy to see that $\theta$ decreases as $\gamma$ increases, and more specifically that $\mathrm{lim}_{\gamma\rightarrow \infty} \theta = 0$. Therefore, we can choose $\gamma$ large enough (with respect to $\varepsilon$), so that $\perpslab{c_i^1c_i^2}$ intersects $t_j$ if and only if $|a_i-b_j| < \varepsilon$, which, by construction, happens only when $a_i=b_j$.  The result follows.
\end{proof}

\ignore{To prove this, we build an embedding of a path in $\mathbb{R}^3$ using `cannons' and `targets', where a slab perpendicular to a `cannon' collides with a `target' if and only if the corresponding elements of the sets $A$ and $B$ are identical.}

The same construction also yields the following:
\begin{corollary}
Given a straight-line drawing of a path $P=(v_1,v_2,\ldots,v_n)$ in $\mathbb{R}^3$, at least $\Omega(n\log n)$ time is required in the algebraic computation tree model to test whether $P$ has increasing chords.
\end{corollary}


\section{Finding self-approaching paths in graphs}
\label{sec:SAPathsInGraphs}

{\changed We do not know how to test in polynomial time if a given graph drawing is self-approaching.
This contrasts with the situation for greedy drawings where it suffices to find, for every pair of vertices $s$ and $t$, a ``first edge'' $(s,a)$ with $D(a,t) < D(s,t)$.
In this section we explore the problem of finding a self-approaching path between two vertices $s$ and $t$ in a graph drawing.  If we could do this in polynomial time, then we could test if a drawing is self-approaching in polynomial time.
We are unable to settle the complexity in two dimensions, but,  by employing the cannons and targets introduced in Section~\ref{sec:testingPaths}, we can show that the problem is hard in three or more dimensions:}


\begin{theorem}
\label{thm:SAPathsInGraphs}
Given a straight-line drawing of a graph $G$ in $\mathbb{R}^3$, and a pair of vertices $s$ and $t$ from $G$, it is NP-hard to determine if a self-approaching $st$-path exists. It is also NP-hard to determine if an increasing-chord $st$-path exists.
\end{theorem}
\begin{proof}
We establish the result for the case of self-approaching paths; the proof for the increasing-chord case is similar.  We reduce from 3SAT.  Let $\mathcal{I}$ be an instance of 3SAT. Let $\{x_1,x_2,\ldots,x_n\}$ be the variables in $\mathcal{I}$. For any $1\leq k \leq n$, let the literal $y_k$ be the negation of the literal $z_k$, both associated with the boolean variable $x_k$. Let $\{w_1,w_2,\ldots,w_m\}$ be the set of clauses associated with $\mathcal{I}$, where $w_i=\{w_i^1,w_i^2,w_i^3\}$ and each literal $w_i^j$ is either $y_k$ or $z_k$ for some value of $k$.  Let $\varepsilon = \pi/2n$. We draw the graph $G$ as follows:
\begin{enumerate}
\item Place the vertex $s$ at the origin.
\item Place two cannons $c_1$ and $c_2$ corresponding to $y_1$ and $z_1$, both at $s$.
\item For all $1 < i \leq n$, place two cannons $c_{2i-1}$ and $c_{2i}$ corresponding to $y_i$ and $z_i$, both at the point $c_{2i-2}^2$ = $c_{2i-3}^2$.
\item Place a vertex $s'$ at $(\alpha n+\gamma,0)$, adjacent to $c_{2n}^2$.
\item Place three targets $t_1$, $t_2$ and $t_3$ at $s'$ with respect to the line $\overline{sc_1^1}$.
\item For all $1 \leq i \leq m$, place three targets $t_{3i-2}$, $t_{3i-1}$ and $t_{3i}$ at $t_{3i-3}^2$, with respect to the line $\overline{sc_1^1}$.
\item For all $1 \leq i \leq 2n$, rotate $c_i^1$ about $x$-axis through an angle of $i\varepsilon$.
\item For all $1 \leq i \leq m$ and $1\leq j \leq 3$, suppose that $w_i^j = y_k$ (respectively, $z_k$).  Then rotate $t_{3(i-1)+j}^1$ about the $x$-axis through an angle of $(2k-1)\varepsilon$ (respectively, $2k\varepsilon$)---in other words, rotate $t_{3(i-1)+j}^1$ through the same amount that the cannon corresponding to the value of the literal $w_i^j$ is rotated, so that a cannon `hits' a target if and only if the cannon and target correspond to the same literal.
\end{enumerate}
The rest of the proof is quite similar to the proof of Lemma~\ref{thm:hardnesspath}.  In particular, we shall show that $\mathcal{I}$ is satisfiable if and only if there is a self-approaching path from $s$ to $t_{3m}^2$.  We will reuse the following statement from the proof of Lemma~\ref{thm:hardnesspath}: for $1\leq i\leq n$, $\perpslab{c_i^1c_i^2}$ intersects the target $t_j$, if and only if $t_j^1$ and $c_i^1$ are rotated by the same amount, hence correspond to the same literal. Let $P$ be a path from $s$ to $t_{3m}^2$. Assume $P$ is a self-approaching path. For each cannon $c_i$ appearing in $P$, assign the literal corresponding to $c_i$ to be false, and its negation to be true. Then, it is easy to show that in each clause, there is at least one true literal: the one appearing in $P$. Similar to this, from a satisfying assignment of the variables, we can construct a self-approaching path by taking the cannons corresponding to false literals. For the second part of the path, we use one of the three targets assigned to each clause: one that corresponds to a true literal. This way, since each target that is traversed in $P$ corresponds to a cannon that was not traversed in $P$, $P$ would be a self-approaching path.

The same proof also works to establish NP-hardness for finding an increasing chord $st$-path.  Note that this is because the drawing of the graph is constructed in a way that any increasing-chord path connecting $s$ to $t_{3m}^2$ is a self-approaching path in the $s$-to-$t_{3m}^2$ direction and vice versa.
\end{proof}

\ignore{To prove this theorem, we reduce from 3SAT.  Our proof uses similar `cannons' and `targets' to those used in the proof of Theorem~\ref{thm:hardnesspath}, but this time, the cannons correspond to variable assignments and the targets correspond to literals in clauses.}


\section{Recognizing graphs having self-approaching drawings}
\label{sec:SADrawability}

In this section
{\changed we characterize trees that have self-approaching drawings and give a linear time recognition algorithm.
This is similar to Moitra's characterization of trees that admit greedy drawings~\cite{Moitra:Thesis:2009}.}  We begin with a simple observation about self-approaching drawings of trees.

\begin{lemma}
\label{lem:treeStrip}
In a self-approaching drawing of a tree $T$, for each edge $(u,v)$, there is no edge or vertex of $T \setminus {uv}$ that intersects $\perpslab{uv}$.
\end{lemma}
\begin{proof}
Since there is a unique path connecting vertices $s$ and $t$ in any tree $T$, a drawing of $T$ is self-approaching if and only if it has increasing chords.  The result then follows from Corollary~\ref{lem:perpendIC}.
\end{proof}


With this lemma in hand, we state the main theorem of this section.

\begin{theorem}
\label{thm:drawableTrees}
Given a tree $T$, we can decide in linear time whether or not $T$ admits a self-approaching drawing.
\end{theorem}
\begin{proof}
To prove this theorem, we completely characterize trees that admit self-approaching drawings.  We require two definitions of special graphs.

A \emph{windmill} having \emph{sweep length} $k$ is a tree constructed by subdividing
each edge of $K_{1,3}$ with $k-1$ new vertices
and then attaching a leaf to each subdivision vertex.
\changeAL{The three subgraphs formed by removing the central vertex of the original $K_{1,3}$ are called \emph{sweeps} and the path of $k$ vertices in each sweep is called the \emph{shaft}.}
A windmill is depicted in Figure~\ref{fig:windmillcrab}(a).

\begin{figure}
\begin{center}
\hspace*{\fill}
\subfloat[]{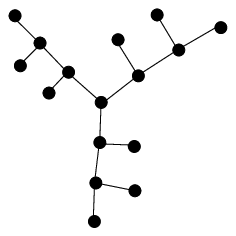}
\hspace*{\fill}
\subfloat[]{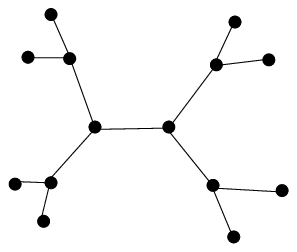}
\hspace*{\fill}
\end{center}
 \caption{(a) A windmill with sweeps of \changeAL{length $3$} and (b) the crab graph. 
 } \label{fig:windmillcrab}
\end{figure}

The \emph{crab graph} is the 14-vertex tree depicted in Figure \ref{fig:windmillcrab}(b).
A graph $G$ is \emph{crab-free} if it has no subgraph that is isomorphic to some subdivision of the crab graph.

We prove Theorem~\ref{thm:drawableTrees} in two steps.  Write $\Delta_T$ for the maximum degree of a vertex in $T$.
\begin{enumerate}
\item First we show that a tree $T$ with $\Delta_T \geq 4$ admits a self-approaching drawing if and only if $T$ is a subdivision of $K_{1,4}$.
\item Then we show that a tree $T$ with $\Delta_T \leq 3$ admits a self-approaching drawing if and only if it is a
\changeAL{subgraph of a} subdivision of a windmill, which happens if and only if $T$ is crab-free.
\end{enumerate}

To establish the first result, the following can be proved:
\begin{lemma}
\label{pathdegree}
In an increasing-chord drawing of a path, the sum of the sizes of the angles in any consecutive chain of $k$ left turns (or right turns) is at least $\pi(k-1)$ if $k>1$ and at least $\pi/2$ if $k=1$.
\end{lemma}
\begin{proof}
There is clearly no angle smaller than $\pi/2$ in any increasing-chord drawing of a path. Let $(u',u)$ and $(v,v')$ be the first and last edges of the chain. Let $s$ be the point in the plane such that
$\angle{uu's}$  
and
$\angle{vv's}$  
are right angles (See Figure~\ref{fig:anglesgame}).
\changeA{Suppose without loss of generality that $s$ lies to the left of the chain.
The path plus $s$ forms a simple counterclockwise polygon of $k+3$ vertices because $l_{uu'}$ and $l_{vv'}$ do not intersect the $uv$-path.
For the same reason,
angle $\angle{u'sv'}$ is less than $\pi$.
The sum of the internal angles of a simple polygon on $n$ vertices is $\pi(n-2)$.  Thus the sum of the angles on the left of the vertices along the $uv$-path is $\pi(k+1) - 2\pi/2 - \angle{u'sv'} \ge \pi(k-1)$.
To argue about the right side angles, note that the sum of the external angles of a simple polygon on $n$ vertices is $\pi(n+2)$.  Also the exterior angle at $s$ is at most $2\pi$.  Thus the sum of the angles on the right of the vertices along the $uv$ path is at least $\pi(k+5) - 2(3\pi/2) -2\pi = \pi k$.
}

\begin{figure}
\begin{center}
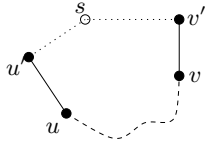
\caption{For proof of Lemma~\ref{pathdegree}.}
\label{fig:anglesgame}
\end{center}
\end{figure}
\end{proof}

\begin{corollary}
If $T$ admits a self-approaching drawing, then $\Delta_T \leq 4$. Also, if $\Delta_T = 4$, then there is only one vertex of degree $4$ in $T$, and the four angles at the vertex of degree $4$ all have size $\pi/2$, and the rest of the angles have size $\pi$.
\end{corollary}
This concludes the first step of the proof.  For the second step, we prove the following three structural lemmas, which establish the equivalence of a tree being a subdivision of a windmill, being crab-free, and admitting a self-approaching drawing.

\begin{lemma}
\label{onethreelemma}
Let $T$ be a crab-free tree with $\Delta_T \leq 3$. Then $T$ is \changeS{a subgraph of} a subdivision of a windmill.
\end{lemma}
\begin{proof}
\changeS{We say that a degree-3 vertex $s$ is \emph{canonical} if there are three disjoint paths connecting $s$ to other degree-3 vertices. For example, vertices $a$ and $b$ in Figure \ref{fig:windmillcrab}(b) are canonical. To prove the lemma we look at three cases: (a) there are two or more canonical vertices; (b) there are no canonical vertices; and (c) there is exactly one canonical vertex.

\changeAL{
a) We rule out this case by showing that if $T$ has two canonical degree-3 vertices $a$ and $b$ then it contains a subgraph that is isomorphic to the crab graph:
In the subgraph formed by deleting the $ab$ path there are
two degree-3 vertices $a_1$ and $a_2$ that have disjoint paths to $a$, and two degree-3 vertices $b_1$ and $b_2$ that have disjoint
paths to $b$.
Now it is easy to see that the minimal connected subgraph of $T$ that contains the vertices $a_1,a_2,b_1,b_2,a,b$ and their neighbours is isomorphic to a subdivision of the crab graph.
}

b) If there are no canonical vertices, then there is a path in $T$ that contains all degree $3$ vertices. Such a graph is isomorphic to a subdivision of a sweep which is a subgraph of the windmill.

c) Now it remains to show that the lemma holds if there is a single canonical vertex $s$ in $T$. Suppose $T$ is rooted at $s$ which has three children. If we remove the subtrees rooted at any two children of $s$, we are left with a graph with no canonical vertices. As we showed, such a graph is isomorphic to a subdivision of a sweep. Furthermore, $s$ is an end vertex of the sweep.  This gives us a way to decompose $T$ into three subgraphs intersecting at $s$, such that each subgraph is a subdivision of a sweep, constituting a windmill.}
\end{proof}

\begin{lemma}
\label{secondlemma}
Let $T$ be a tree that is a subdivision of a windmill.  Then $T$ admits a self-approaching drawing.
\end{lemma}
\begin{proof}
It suffices to show that any windmill admits a self-approaching drawing. We draw a $K_{1,3}$ so that each angle is $2\pi/3$ and edges are unit length.
\changeAL{From each leaf $l$, draw two rays  so that the  wedge between them has angle
$\pi/2 + \varepsilon$
for some small $\varepsilon$ and each of the angles formed by a ray and the incident edge of the $K_{1,3}$ is
$3\pi/4 - \varepsilon/2$.
It can easily be seen that \changeS{for small enough $\varepsilon$,
if we expand the wedge at $l$ by $\pi/2$ on each side then this ``wide'' wedge of angle $3\pi/2 + \varepsilon$ does not contain any part of the drawing of $K_{1,3}$ (See Figure \ref{fig:rays}).
In fact the distance of each of the two other leaves to this wedge is at least
$\sin(\pi/4-\varepsilon/2-\pi/6)$.
}}

\begin{figure}
\begin{center}
\includegraphics[width=3.5in]{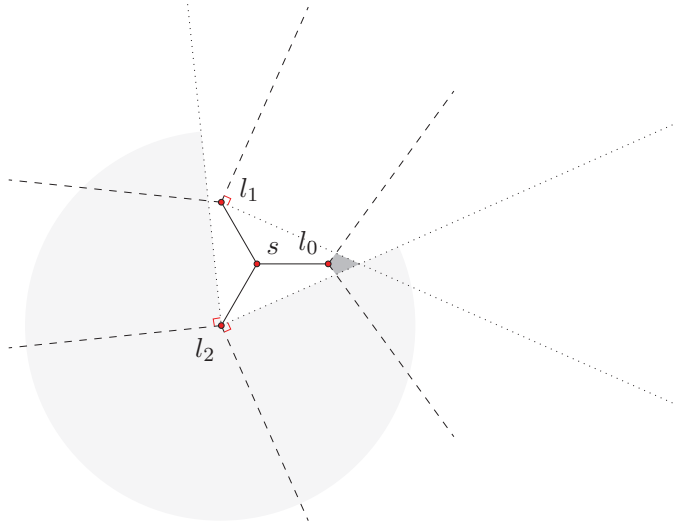}
\caption{
\changeAL{Self-approaching drawing of a windmill: The drawing of $s$ and its three neighbors (solid lines) along with the two rays at each of the neighbors (dashed). 
The wide wedge at $l_2$ is lightly shaded.  The sweep containing $l_0$ will be drawn in the darkly shaded region between the two rays at $l_0$ and outside the wide wedges at $l_1$ and $l_2$.}
}
\label{fig:rays}
\end{center}
\end{figure}

\begin{figure}
\begin{center}
\includegraphics[width=5in]{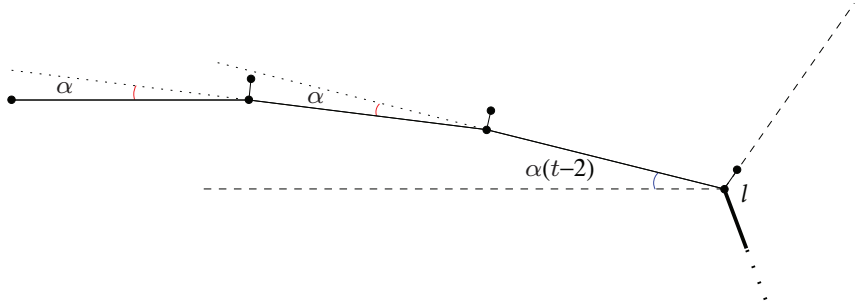}
\caption{Self-approaching drawing of a windmill: Drawing a sweep of \changeAL{length $t=4$}. The two rays are drawn using dashed segments and
\changeAL{$\alpha$ here is $\varepsilon/2(t-2)$.}}
\label{windmillembeddfig2}
\end{center}
\end{figure}

Let $\gamma$ be a number to be set later. For each leaf $l$ of the drawing of $K_{1,3}$, we draw the sweep that includes $l$ as follows. \changeAL{Assume that $l$ is part of a sweep of length $t$. We draw the sweep between the two rays at $l$ and outside the wide wedges of the other two leaves.
Furthermore,  we ensure that  the strip $l_e$ of each edge $e$ of the sweep lies inside the wide wedge at $l$.  This prevents intersections between strips of one sweep and edges of any other sweep.

We first draw the shaft of the sweep.
Draw the first edge incident to $l$ so that it has length $\gamma$ and makes an angle of  $\varepsilon/2$  with one of the rays at $l$. Continue to draw the rest of the shaft with each edge having a $\frac{\varepsilon}{2(t-2)}$ difference of direction with the previous edge and length $\gamma$ (See Figure~\ref{windmillembeddfig2}). This means that the last edge of the shaft is parallel to one of the two rays at $l$.
To ensure that the drawing stays outside the other wide wedges,
$\gamma$ can be set to $\sin(\pi/4-\varepsilon/2-\pi/6)/t$.


Next we draw the leaves of the sweep.
Draw the leaf attached to $l$ so that it is inside the reflex angle at $l$ and lies exactly on one of the rays. Then draw the rest of the leaves 
in such a way that each new edge is exactly in the middle of the reflex angle of the two incident edges of the shaft (See Figure~\ref{windmillembeddfig2}). The length of each of these new edges should be small enough so that none of them is inside the strip induced by another one. To satisfy this, the length of each such leaf can be \changeS{$\gamma\tan (\varepsilon/4t)$.}
Note that the strip of each of these edges lies inside the wide wedge at $l$.
}
\end{proof}

\begin{lemma}
\label{subcrablemma}
Let $T$ be a tree that contains a subdivision of the crab.  Then $T$ does not admit a self-approaching drawing.
\end{lemma}
\begin{proof}
It is easy to see that if a tree admits a self-approaching drawing, then any connected subgraph of it also admits a self-approaching drawing. Therefore, we only need to show that no subdivision of the crab graph has a self-approaching drawing. \changeS{First we show that the crab graph itself does not admit a self-approaching drawing.} By Lemma \ref{pathdegree}, the total size of the chain of four angles on the path from $a_{1,2}$ to $b_{1,1}$ is greater than $3\pi$. By similar arguments, the angles on the path from $a_{22}$ to $b_{22}$ also sum to $3\pi$. Similarly, by Lemma \ref{pathdegree}, the total size of the chain of three consecutive angles on the path from $a_{1,1}$ to $a_{2,1}$ is greater than $2\pi$. By similar arguments, the angles on the path from $b_{12}$ to $b_{21}$ also sum to $2\pi$. By Lemma \ref{pathdegree}, each of the four angles formed by the eight leaves has size at least $\pi/2$, summing to $2\pi$. This adds up to a total strictly greater than $3\pi+3\pi+2\pi+2\pi+2\pi = 12\pi$. Since these angles are the angles around the $6$ vertices $a, b, a_{1},a_{2},b_{1}$, and $b_{2}$, we have a contradiction.

\changeS{Now consider $C$ to be a subdivision of the crab graph. Each subdivision vertex adds a total of $2\pi$ to the both sides of the inequality, hence the contradiction holds.}
\end{proof}

Combining these results, we obtain the second step of the proof of the theorem.  This completes the characterization of all trees that admit self-approaching drawings.  To complete the proof of Theorem~\ref{thm:drawableTrees}, it suffices to observe that it is possible, in linear time, to check whether a tree $T$ is a subdivision of $K_{1,4}$ or of a windmill.
\end{proof}


\section{Constructing self-approaching Steiner networks}
\label{sec:SAspanners}
We now turn our attention to the following problem: Given a set $P$ of points in the plane, draw a graph $N$ with straight edges and $P\subseteq V(N)$ such that for each ordered pair of points $p,q\in P$ there is a self-approaching path from $p$ to $q$ in the drawing of $N$. We call the points in $V(N)\backslash P$  \emph{Steiner points} and the graph $N$ a \emph{self-approaching Steiner network for $P$}. An increasing-chord Steiner network is defined similarly.

We show that small increasing-chord Steiner networks (and hence small self-approaching Steiner networks) can always be constructed for any given set of points in the plane.

\begin{theorem}
\label{thm:steiner}
Given a set $P$ of $n$ points in the plane, there exists an increasing-chord Steiner network for $P$ having $O(n)$ vertices and edges.
\end{theorem}
\begin{proof}
Given points $p$ and $q$, let $\theta_{pq}$ denote the angle between the
line $pq$ and the $x$-axis (we take the smaller of the two angles
formed, so that $\theta_{pq}\in [0,\pi/2]$).
A path is {\em $xy$-monotone\/} if every vertical line intersects the path
\changeA{in at most one point or one segment and every horizontal line intersects the path in at most one point or one segment.}
Clearly, an $xy$-monotone path is self-approaching.
\changeA{We will use rectilinear $xy$-monotone paths in our construction.}
We will build a linear-size Steiner network~$G$
with the following property:
\begin{quote}
For every pair of points $p,q\in P$ with
$\theta_{pq}\in [\pi/8,3\pi/8]$, there is
a \changeA{rectilinear} $xy$-monotone path from $p$ to $q$ in~$G$.
\end{quote}
To handle the remaining pairs of points, we can rotate the coordinate axes by $\pi/4$
and apply the same construction to obtain another Steiner network $G'$.
We can then return the union of $G$ and~$G'$.

To construct $G$, we first build a {\em quadtree}~\cite{Har-Peled:book},
defined as follows:
The root stores an initial square enclosing $P$.  At each node,
we divide its square into four congruent subsquares and create
a child for each subsquare that is not empty of points of $P$.
The tree has $n$ leaves.

To ensure that the tree has $O(n)$ internal nodes, we compress each maximal path
of degree-1 nodes by keeping only the first and last node in the path.
The result is a {\em compressed quadtree}, denoted $T$.

For each square $B$ in the compressed quadtree $T$, we add the four corner vertices
and edges of $B$ to $G$.
(Note that we allow overlapping edges in our construction;
it is not difficult to avoid overlaps by subdividing the edges appropriately.)
For each leaf square $B$ in $T$ containing a single point $p\in P$,
we add a 2-link $xy$-monotone path in $G$ from $p$ to each corner of $B$.
For each degree-1 square $B$ in $T$ having a single child square $B'$,
we add a 2-link $xy$-monotone path in $G$
from each corner of $B'$ to the corresponding corner of $B$.
By induction, it then follows
that for every point $p\in P$ inside a square $B$ in $T$,
there is an $xy$-monotone path in $G$ from $p$ to each corner of $B$.
The number of vertices and edges in $G$ thus far is $O(n)$.

Given a parameter $\varepsilon>0$,
a {\em well-separated pair decomposition\/} of $P$ is
a collection of pairs of sets $\{A_1,B_1\},\ldots,\{A_s,B_s\}$,
such that\footnote{
In the original definition~\cite{CalKos}, $A_i$ and $B_i$ are subsets of $P$,
but for our purposes, we will take $A_i$ and $B_i$ to be regions in the plane
(namely, squares).
}
\begin{enumerate}
\item for every pair of points $p,q\in P$, there is a unique index $i$
with $(p,q)\in A_i\times B_i$ or $(p,q)\in B_i\times A_i$;
\item $A_i$ and $B_i$ are {\em well-separated\/} in the sense that
both the diameter of $A_i$ and the diameter of $B_i$ is at most
$\varepsilon d(A_i,B_i)$, where $d(A_i,B_i)$ is the minimum distance between $A_i$ and $B_i$.
\end{enumerate}
It is known that a well-separated pair decomposition consisting
of $s=O(n/\varepsilon^2)$ pairs always exists~\cite{CalKos}.
Furthermore, such a decomposition
can be constructed by a simple quadtree-based algorithm (for example, see
\cite{Har-Peled:book} or \cite{Chan:wspd}), where the sets $A_i$ and $B_i$
are in fact squares appearing in the compressed quadtree $T$.

To finish the construction of $G$, we consider each pair $\{A_i,B_i\}$ in the
decomposition such that $A_i$ and $B_i$ are separated by both a vertical line
and a horizontal line.
Without loss of generality, suppose that $A_i$ is to the left of and below $B_i$.
We add a 2-link $xy$-monotone path in $G$ from the upper right
corner of $A_i$ to the lower left corner of $B_i$.
The overall number of vertices and edges in $G$ is $O(n/\varepsilon^2)$.

To show that $G$ satisfies the stated property, let $p,q\in P$
with $\theta_{pq}\in [\pi/8,3\pi/8]$.
Suppose that $(p,q)\in A_i\times B_i$.
If $A_i$ and $B_i$ are intersected by a common horizontal line, then
$\theta_{pq}$ must be upper-bounded by $O(\varepsilon)$ because $A_i$ and $B_i$
are well-separated; this is a contradiction
if we make the constant $\varepsilon$ sufficiently small.
Thus, $A_i$ and $B_i$ must be separated by a horizontal line, and similarly
by a vertical line via a symmetric argument.
Without loss of generality, suppose that $A_i$ is to the left of and below $B_i$.
By concatenating $xy$-monotone paths in $G$, we can get from $p$
to the upper right corner of $A_i$, then to the lower left corner of $B_i$,
and finally to $q$.
\end{proof}

In the above construction, the edges
we add for each well-separated pair $\{A_i,B_i\}$ may cross other edges, although
it is possible to modify the construction to
ensure that the network $G$ is planar (and similarly $G'$).  However, we do not
know how to avoid crossings in the final network
obtained by unioning $G$ and $G'$, while keeping the number of edges linear.
Our construction can be carried out in $O(n \log n)$ time, since that is the cost for building the compressed quad tree and the  well-separated pair decomposition.  The theorem generalizes to any constant dimension.

We note that our construction bears some similarity to the construction used independently by Borradaile and Eppstein~\cite{Borradaile} to create small low-weight plane Steiner spanners in which the paths stay within a bounded range of angles.


Whether planar self-approaching Steiner networks of linear size can be constructed or not is an interesting question.
Delaunay triangulations seemed to be a potential candidate, however, Figure~\ref{fig:Del-not-SA} shows a configuration of 6 points in the plane whose Delaunay triangulation is not a self-approaching drawing.

\begin{figure}
\begin{center}
\includegraphics[width=0.5\textwidth]{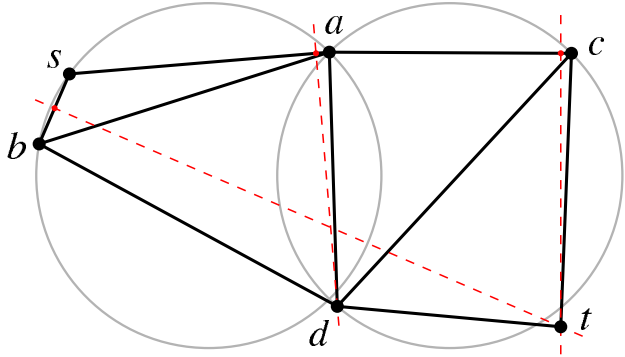}
\vspace{1cm}
\caption{The Delaunay triangulation of these six points does not have a self-approaching path from $s$ to $t$.  Forbidden edge-vertex pairs are indicated with dashed lines.  From $s$ we must take edge $sa$, because $t$ lies in the forbidden region for edge $sb$.   Then we cannot go to $d$ since it is in the forbidden region of $sa$, nor can we use edge $ac$ since $t$ is in its forbidden region.}
\label{fig:Del-not-SA}
\end{center}
\end{figure}

\section{Conclusions}

We have introduced the notion of self-approaching and increasing-chord graph drawings, with rich connections to greedy drawings, spanners, dilation and detour, and minimum Manhattan networks.

Our results are preliminary.  We leave open the following questions:
\begin{itemize}
\item Can we test, in polynomial time, if a straight-line graph drawing in the plane is self-approaching [or increasing-chord]? Or is the problem NP-complete?
{\changed
\item Given a graph $G$, can we efficiently produce a self-approaching drawing of $G$ if one exists?

\item What classes of graphs have self-approaching [or increasing-chord] drawings? Does, for example, every 3-connected planar graph have a self-approaching drawing?  Even more interesting, which graphs have a self-approaching drawing such that
local routing finds a self-approaching path?
For example, if 3-connected graphs had such drawings, this}
would have the significant implication that every 3-connected planar graph has an embedding where local routing gives paths of bounded detour (hence bounded dilation).
Bose \etal\cite{Bose:theta6:2012} recently proved the weaker result that every triangulation has an embedding where local routing gives paths of bounded dilation.
\end{itemize}


\medskip\noindent
{\bf Acknowledgements.}  Anna Lubiw would like to thank Marcus Brazil, Victor Chepoi, Matthias M\"uller-Hannemann, and Martin Zachariasen for Dagstuhl workshop discussions that inspired this line of enquiry. This work was done as part of an Algorithms Problem Session at the University of Waterloo, and we thank the other participants for helpful discussions.  We thank Prosenjit Bose and Pat Morin for help finding the example in Figure~\ref{fig:Del-not-SA}.

\bibliographystyle{abbrv}
\bibliography{SAdrawings}

\end{document}